\documentclass{stacs_proc}

\usepackage{url,amsmath,amssymb,enumerate,xspace}

\renewcommand{\bar}{\overline}

\newcommand{\cA}{\mathcal A}
\newcommand{\ch}{{\mathrm{ch}}} 
\newcommand{\cl}{{\mathrm{closed}}} 
\newcommand{\clone}{{\mathrm{cl}}} 
\newcommand{\cM}{\mathcal M}
\newcommand{\cN}{\mathcal N}

\newcommand{\down}{{\mathord\downarrow}}
\newcommand{\Down}[1]{\down #1\setminus\{#1\}}
\newcommand{\fA}{\mathfrak A}

\newcommand{\fB}{\mathfrak B}

\newcommand{\FO}{\mathrm{FO}}

\renewcommand{\hat}{\widehat}

\newcommand{\mch}{{\mathrm{mch}}} 
\newcommand{\MSO}{{\mathrm{MSO}}}
\newcommand{\N}{\mathbb N}
\newcommand{\dkpath}{{\mathrm{reach}}} 

\newcommand{\red}{\mathrm{red}}

\newcommand{\w}{{\mathrm{w}}} 

\begin{document}
\allowdisplaybreaks{}

\title[Iterations and monadic second order logic]{Compatibility of
  Shelah and Stupp's and Muchnik's iteration with fragments
  of monadic second order logic}
\author{Dietrich Kuske}{Dietrich Kuske}
\address{Institut f\"ur Informatik, Universit\"at Leipzig}

\begin{abstract}
  We investigate the relation between the theory of the iterations in
  the sense of Shelah-Stupp and of Muchnik, resp., and the theory of
  the base structure for several logics. These logics are obtained
  from the restriction of set quantification in monadic second order
  logic to certain subsets like, e.g., finite sets, chains, and finite
  unions of chains. We show that these theories of the Shelah-Stupp
  iteration can be reduced to corresponding theories of the base
  structure. This fails for Muchnik's iteration.
\end{abstract}


\keywords{Logic in computer science, Rabin's tree theorem}
\subjclass{F.4.1}


\maketitle

\stacsheading{2008}{467-478}{Bordeaux}
\firstpageno{467}

\section{Introduction}

Rabin's tree theorem states, via an automata-theoretic proof, the
decidability of the monadic second order (short: $\MSO$) theory of the
complete binary tree. It allows to derive the decidability of
seemingly very different theories (e.g., the $\MSO$-theory of the real
line where set quantification is restricted to closed
sets~\cite{Rab69}).  Its importance is stressed by Seese's result that
any class of graphs of bounded degree with a decidable $\MSO$-theory
has bounded tree-width (i.e., is ``tree-like'')~\cite{See91}.

In \cite{She75}, Shelah reports a generalization of Rabin's tree
theorem that was proved by Shelah and Stupp. The idea is to start with
a structure $\fA$ and to consider the tree whose nodes are the finite
words over the universe of~$\fA$ together with the prefix order on
these words. Then the immediate successors of any node in this tree
can naturally be identified with the elements of the structure $\fA$
-- hence they carry the relations of $\fA$. The resulting tree with
additional relations is called \emph{Shelah-Stupp-iteration}. The
above mentioned result of Shelah and Stupp states that the
$\MSO$-theory of the Shelah-Stupp-iteration can be reduced to the
$\MSO$-theory of the base structure~$\fA$. If $\fA$ is the
two-elements set, then Rabin's tree theorem follows. 

A further extension is attributed to Muchnik~\cite{Sem84} who added a
unary clone predicate to Shelah and Stupp's iteration resulting in the
\emph{Muchnik-iteration}. This clone predicate states that the last
two letters of a word are the same. This allows, e.g., to define the
unfolding of a rooted graph in its Muchnik-iteration~\cite{CouW98}.
Muchnik's theorem then gives a reduction of the MSO-theory of the
Muchnik-iteration to the MSO-theory of the base structure. The proof
was not published by Muchnik himself, but, using automata-theoretic
methods, Walukiewicz showed that the reduction in Muchnik's theorem is
even uniform (i.e., independent from the concrete base
structure)~\cite{Wal02}. Since, as mentioned above, the unfolding of a
rooted graph can be defined in the Muchnik-iteration, the MSO-theory
of this unfolding can be reduced to that of the
graph~\cite{CouW98}. This result forms the basis for Caucal's
hierarchy~\cite{Cau02b} of infinite graphs with a decidable
MSO-theory. Walukiewicz's automata-theoretic proof ideas have been
shown to work for the Muchnik-iteration and stronger logics like
Courcelle's counting MSO and guarded second-order logic by Blumensath
\& Kreutzer~\cite{BluK05}.

In \cite{KusL06a}, we asked for a first-order version of Muchnik's
result -- and failed. More precisely, we constructed structures with a
decidable first-order theory whose Muchnik-iteration has an
undecidable first-order theory. As it turns out, the only culprit is
Muchnik's clone predicate since, on the positive side, we were able to
uniformly reduce the first-order theory (and even the monadic chain
theory where set variables range over chains, only) of the
Shelah-Stupp-iteration to the first-order theory of the base
structure.\footnote{In the meantime, Alexis Bes found a simpler proof
  of a stronger result based on the ideas of automatic structures and
  \cite{Tho87} (personal communication).} 

The aim of this paper is to clarify the role of weak monadic second
order logic~$\MSO^\w$ in the context of Shelah-Stupp- and
Muchnik-iteration.  We first define infinitary versions of these
iterations that contain, in addition to the finite words, also
$\omega$-words. On the positive side, we prove a rather satisfactory
relation between the theories of the infinitary Shelah-Stupp-iteration
and the base structure.  More precisely, the Shelah-Stupp result
together with some techniques from~\cite{Rab69} allows to uniformly
reduce the $\MSO^\cl$-theory of the infinitary Shelah-Stupp-iteration
(where set quantification is restricted to closed sets) to the
MSO-theory of the base set. Our result from \cite{KusL06a} ensures
that Shelah-Stupp-iteration is $\FO$-compatible in the sense of
Courcelle (i.e., the $\FO$-theory of the infinitary
Shelah-Stupp-iteration can be reduced uniformly to the $\FO$-theory of
the base structure). Our new positive result states that
Shelah-Stupp-iteration is also $\MSO^\w$-compatible. To obtain this
result, one first observes that the finiteness of a set in the
Shelah-Stupp-iteration is definable in $\MSO^\mch$ (where
quantification is restricted to finite unions of chains), hence the
$\MSO^\w$-theory of the Shelah-Stupp-iteration can be reduced to its
$\MSO^\mch$-theory.  For this logic, we then prove a result analogous
to Rabin's basis theorem: Any consistent $\MSO^\mch$-property in the
Shelah-Stupp-iteration of a finite union of chains (i.e., of a certain
set of words over the base structure) has a witness that can be
accepted by a small automaton. But an automaton over a fixed set of
states can be identified with its transition matrix, i.e., with a
fixed number of finite sets in the base structure. We then prove that
$\MSO^\mch$-properties of the language of an automaton can effectively
be translated into $\MSO^\w$-properties of the transition matrix.

On the negative side, we prove that infinitary Muchnik-iteration is
not $\MSO^\w$-com\-pat\-i\-ble. Namely, there is a tree $T_\omega$
with decidable $\MSO^\w$-theory such that for any set $M$ of natural
numbers, there exists an $\MSO^\w$-equivalent tree $\fA_M$ such that
$M$ can be reduced to the $\MSO^\w$-theory of the infinitary
Muchnik-iteration of~$\fA_M$. This proof uses the fact that the
existence of an infinite branch in a tree is not expressible
in~$\MSO^\w$, but it is a first-order (and therefore a $\MSO^\w$-)
property of the infinitary Muchnik-iteration.

\section{Preliminaries}

\subsection{Logics}

A \emph{(relational) signature} $\sigma$ consists of finitely many
constant and relation symbols (together with the arity of the latter);
a \emph{purely relational signature} does not contain any constant
symbols. Formulas use \emph{individual} and \emph{set variables},
usually denoted by small and capital, resp., letters from the end of
the alphabet.  \emph{Atomic formulas} are $x_1=x_2$,
$R(x_1,\dots,x_n)$, and $x_1\in X$ where $R$ is an $n$-ary relation
symbol from $\sigma$, $x_1,x_2,\dots,x_n$ are individual variables or
constant symbols, and $X$ is a set variable. \emph{Formulas} are
obtained from atomic formulas by conjunction, negation, and
quantification $\exists Z$ for $Z$ an individual or a set variable. A
\emph{sentence} is a formula without free variables. The satisfaction
relation $\models$ between a $\sigma$-structure~$\fA$ and formulas is
defined as usual. For two $\sigma$-structures $\fA$ and $\fB$, we
write $\fA\equiv^\MSO_m\fB$ if, for any sentence $\varphi$ of
quantifier depth at most~$m$, we have $\fA\models\varphi$ iff
$\fB\models\varphi$.  If $\fA$ and $\fB$ agree on all first-order
formulas (i.e., formulas without set quantification) of quantifier
depth at most $m$, then we write $\fA\equiv^\FO_m\fB$.

Let $(V,\preceq)$ be a partially ordered set. A set $M\subseteq V$ is
a \emph{chain} if $(M,\preceq)$ is linearly ordered, it is a
\emph{multichain} if $M$ is a finite union of chains. An element $x\in
M$ is a \emph{branching point} if $\{y\in M\mid x<y\}$ is nonempty and
does not have a least element.

We will also consider different restrictions of the satisfaction
relation $\models$ where set variables range over certain subsets,
only. In particular, we will meet the following restrictions. 
\begin{itemize}
\item Set quantification can be restricted to finite sets, i.e., we
  will discuss weak monadic second order logic. The resulting
  satisfaction relation is denoted $\models^\w$ and  the
  equivalence of structures $\equiv^\w_m$. 
\item Set quantification can be restricted to chains (where we assume
  a designated binary relation symbol $\preceq$ in $\sigma$) which
  results in $\models^\ch$ and $\equiv^\ch_m$, cf.\ Thomas~\cite{Tho87}. 
\item $\models^\mch$ etc.\ refer to the restriction of set
  quantification to multichains. 
\item The superscript $\cl$ denotes that set variables range over
  closed sets, only (where we associate a natural topology to any
  $\sigma$-structure), cf.\ Rabin~\cite{Rab69}. 
\end{itemize}

Let $t$ be some transformation of $\sigma$-structures into
$\tau$-structures, e.g., transitive closure.  A very strong relation
between the $\mathcal L$-theory of $\fA$ and the $\mathcal K$-theory
of $t(\fA)$ is the existence of \emph{a single} computable function
$\red$ that reduces the $\mathcal K$-theory of $t(\fA)$ to the
$\mathcal L$-theory of $\fA$ \emph{for any
  $\sigma$-structure~$\fA$}. As shorthand for this fact, we say ``The
transformation $t$ is $(\mathcal K,\mathcal L)$-compatible'' or,
slightly less precise ``The $\mathcal K$-theory of $t(\fA)$ is
\emph{uniformly reducible} to the $\mathcal L$-theory of $\fA$.''
$(\mathcal K,\mathcal K)$-compatible transformations are simply called
$\mathcal K$-compatible.

\begin{example}
  Any $\MSO$-transduction is $\MSO$-compatible~\cite{Cou94} and finite
  set interpretations are
  $(\MSO^\w,\FO)$-compatible~\cite{ColL07}. Feferman \& Vaught showed
  that any generalized product is
  FO-compatible~\cite{FefV59}. Finally, any generalized sum is
  $\MSO$-compatible by Shelah~\cite{She75}.
\end{example}

\subsection{Shelah and Stupp's and Muchnik's iteration}

Let $A$ be a (not necessarily finite) alphabet.  With $A^*$ we denote
the set of all finite words over~$A$, $A^\omega$ is the set of
infinite words, and $A^\infty=A^*\cup A^\omega$. The prefix relation
on finite and infinite words is $\preceq$. The set of finite prefixes
of a word $u\in A^\infty$ is denoted $\down u=\{v\in A^*\mid v\preceq
u\}$, if $C\subseteq A^\infty$, then $\down C=\bigcup_{u\in C}\down
u$. For $L\subseteq A^\infty$ and $u\in A^*$ let $u^{-1} L=\{v\in
A^\infty\mid uv\in L\}$ denote the left-quotient of $L$ with respect
to~$u$.

Let $\sigma$ be a relational signature and let
$\fA=(A,(R^\fA)_{R\in\sigma})$ be a structure over the signature
$\sigma$. The \emph{infinitary Shelah-Stupp-iteration $\fA^\infty$ of
  $\fA$} is the structure
\[
  \fA^\infty=(A^\infty,\preceq, (\hat{R})_{R\in\sigma},\varepsilon)
\]
where, for $R\in\sigma$, 
\[
   \hat{R} = \{(ua_1,\dots,ua_n)\mid u\in A^*,
                                         (a_1,\dots,a_n)\in R^\fA\}\ . 
\]
The \emph{(finitary) Shelah-Stupp-iteration $\fA^*$} is the restriction of
$\fA^\infty$ to the set of finite words~$A^*$. 

\begin{example}\label{E-tree}
  Suppose the structure $\fA$ has two elements $a$ and $b$ and two
  unary relations $R_1=\{a\}$ and $R_2=\{b\}$. Then $\hat{R}_1 =
  \{a,b\}^*a$ and $\hat{R}_2 =\{a,b\}^*b$. Hence the finitary
  Shelah-Stupp-iteration~$\fA^*$ can be visualized as a complete
  binary tree with unary predicates telling whether the current node
  is the first or the second son of its father. In addition, the root
  $\varepsilon$ is a constant of the Shelah-Stupp-iteration $\fA^*$.
  Furthermore, the infinitary Shelah-Stupp-iteration $\fA^\infty$ adds
  leaves to this tree at the end of any branch. Since this allows to
  define $(\mathbb R,\le)$ in $\fA^\infty$, the unrestricted
  MSO-theory of $\fA^\infty$ is undecidable.
\end{example}

\newcommand{\Interface}{\mathrm{iface}}

\begin{example}\label{L-monoids} \textbf{(cf.\ \cite{KusL06b})}
  The Shelah-Stupp iteration allows to reduce the Cayley graph of a
  free product to the Cayley graphs of the factors. 
  Let $M_i=(M_i,\circ_i,1_i)$ be monoids finitely generated by
  $\Gamma_i$ for $1\le i\le n$ and let
  $G_i=(M_i,(E_i^a)_{a\in\Gamma_i},\{1_i\})$ denote the rooted Cayley
  graph of $M_i$. Then the Cayley graph
  $G=(P,(E^a)_{a\in\bigcup\Gamma_i})$ of the free product
  $P=(P,\circ,1)$ of these monoids can be defined in the Shelah-Stupp
  iteration of the disjoint union of the Cayley graphs~$G_i$. For this
  to work, let $M=\bigcup_{1\le i\le n}M_i$ be the disjoint union of
  the monoids $M_i$ and consider the structure
  \[
     \cA=(M,(M_i)_{1\le i\le n},
            (E_i^a)_{\substack{1\le i\le n\\a\in\Gamma_i}},U
         )
  \]
  where $U= \{1_i\mid 1\le i\le n\}$ is the set of units.

  Then a word $w\in M^*$ belongs to the direct product $P$ iff the
  following holds in the Shelah-Stupp iteration of $\cA$:
  \[
    \bigwedge_{1\le i\le n}
     \forall x\lessdot y\preceq w:
     x\in\hat{M_i}\rightarrow y\notin\hat{M_i}
      \land y\notin\hat U
  \]
  where $\lessdot$ denotes the immediate successor relation of the
  partial order $\preceq$. For $a\in\Gamma_i$ and $v,w\in P$, we have
  $v\circ a=w$ (i.e., $(v,w)\in E^a$) iff the Shelah-Stupp iteration
  satisfies
\[    \left( \exists v'\in\hat U:v\lessdot v'\land(v',w)\in \hat{E_i^a}\right)
    \lor
      (v,w)\in \hat{E_i^a}
    \lor
     \left(   \exists w'\in\hat U:w\lessdot w'\land(v,w')\in \hat{E_i^a}\right)\ .
\]
\end{example}

Muchnik introduced the additional unary \emph{clone predicate}
$\clone=\{uaa\mid u\in A^*,a\in A\}$. The extension of the
Shelah-Stupp-iterations by this clone predicate will be called
\emph{finitary and infinitary Muchnik-iteration} $(\fA^*,\clone)$ and
$(\fA^\infty,\clone)$, resp. Courcelle and Walukiewicz \cite{CouW98}
showed that the unfolding of a
directed rooted graph~$G$ can be defined in the Muchnik iteration
$(G^*,\clone)$ of~$G$.

To simplify notation, we will occasionally omit the word ``finitary'' and just
speak of the Shelah-Stupp- and Muchnik-iteration. 

\section{A basis theorem for $\MSO^\mch$}

Rabin's tree theorem~\cite{Rab69} states the decidability of the monadic second
order theory of the complete binary tree. As a \emph{corollary} of his proof
technique by tree automata, one obtains Rabin's basis theorem~\cite[Theorem
26]{Rab72}: Let $\varphi$ be a formula with free variables $X_1,\dots,X_\ell$
and let $L_1,\dots,L_\ell\subseteq\{a,b\}^*$ be regular languages such that
the binary tree satisfies $\varphi(L_1,\dots,L_\ell)$. Then it satisfies
$\psi(L_1,\dots,L_\ell)$ where $\psi$ is obtained from $\varphi$ by
restricting all quantifications to regular sets. To obtain this basis theorem,
it suffices to show that validity of $\exists
X_\ell:\varphi(L_1,\dots,L_{\ell-1},X_\ell)$ implies the existence of a
regular set $R_\ell$ such that $\varphi(L_1,\dots,L_{\ell-1},R_\ell)$ holds
true in the binary tree. 

This is precisely what this section shows in our context of
$\MSO^\mch$ and the Shelah-Stupp-iteration~$\fA^*$. Even more, we will
not only show that the set $R_\ell$ can be chosen regular, but we will
also bound the size of the automaton accepting it.

\textit{Throughout this section, $\sigma$ denotes some purely relational
signature.}

\subsection{Preliminaries}

For $k,\ell\in\N$, let $\tau_{k,\ell}$ be the extension of the
signature $(\sigma,\preceq)$ by $k$ constants and $\ell$ unary
relations.  Using Hintikka-formulas (see \cite{EbbF91} for the
definition and properties of these formulas) one can show that for any
of the signatures $\tau_{k,\ell}$ and $m\in\N$, there are only
finitely many equivalence classes of~$\equiv^\mch_m$.  An upper bound
$T(\ell,m)$ for the number of equivalence classes of $\equiv^\mch_m$
on formulas over the signature $\tau_{2,\ell}$ can be computed
effectively.

Now let $\fA=(A,(R)_{R\in\sigma})$ be some $\sigma$-structure.
For $u\in A^*$, let $\fA^*_u$ denote the $\tau_{1,0}$-structure
$(uA^*,\sqsubseteq,(\bar R)_{R\in\sigma},u)$ where
\begin{itemize}
\item the relation $\sqsubseteq$ is the restriction of $\preceq$ to
  $uA^*$ and
\item $\bar R$ is the restriction of $\hat{R}$ to $uA^+$.
\end{itemize}
For any $u,v\in A^*$, the mapping $f:\fA^*_u\to \fA^*_v$ with
$f(ux)=vx$ is an isomorphism -- this is the reason to consider $\bar
R$ and not the restriction of $\hat{R}$ to $uA^*$. Similarly, the
$\tau_{2,0}$-structure $\fA^*_{u,v}=(uA^*\setminus
vA^+,\sqsubseteq,(\bar R)_{R\in\sigma},u,v)$ is defined for $u,v\in
A^*$ with $u\preceq v$. Here, again, $\bar R$ is the restriction of
$\hat{R}$ to $uA^+\setminus vA^+$.  

Frequently, we will consider the structure $\fA^*$ together with some
additional unary predicates $L_1,\dots,L_\ell$. As for the plain
structure $\fA^*$, we will also meet the restriction of
$(\fA^*,L_1,\dots,L_\ell)$ to the set $uA^*$, i.e., the structure
$(\fA^*_u,L_1\cap uA^*,\dots,L_\ell\cap uA^*)$. To simplify notation,
this will be denoted $(\fA^*_u,L_1,\dots,L_\ell)$; the structure
$(\fA^*_{u,v},L_1,\dots,L_\ell)$ is to be understood similarly. 
\medskip

\noindent\textbf{Example~\ref{E-tree} (continued).} 
In the case of Example~\ref{E-tree}, $\fA^*_u$ is just the subtree
rooted at the node~$u$. On the other hand, $\fA^*_{u,v}$ is obtained
from $\fA^*_u$ by deleting all descendants of $v$ and marking the node
$v$ as a constant. Thus, we can think of $\fA^*_{u,v}$ as a tree with
a marked leaf. These \emph{special trees} are fundamental in the work
of Gurevich \& Shelah \cite{GurS83} and of Thomas~\cite{Tho87}.  
\medskip

In the following, fix some $\ell\in\N$. We then define the operations
of product and infinite product of $\tau_{k,\ell}$-structures: If
$\fA=(A,\preceq^\fA,(R^\fA)_{R\in\sigma},a_1,a_2,
L^\fA_1,\dots,L^\fA_\ell)$ is a $\tau_{2,\ell}$-structure and
$\fB=(B,\preceq^\fB,(R^\fB)_{R\in\sigma},b_1,\dots,b_k,
L^\fB_1,\dots,L^\fB_\ell)$ a disjoint $\tau_{k,\ell}$-structure with
$k\ge1$, then their \emph{product} $\fA\cdot\fB$ is a
$\tau_{k,\ell}$-structure. It is obtained from the structure
\[
  (A\cup B,\preceq^\fA\cup\preceq^\fB,(R^\fA\cup R^\fB)_{R\in\sigma},
   L_1^\fA\cup L_2^\fB,\dots,L_\ell^\fA\cup L_\ell^\fB)
\]
by identifying $a_2$ and~$b_1$, taking the transitive closure of the
partial orders, and extending the resulting structure by the list of
constants $a_1,b_2,b_3,\dots,b_k$. Now let $\fA_n$ be disjoint
$\tau_{2,\ell}$-structures with constants $u_n$ and $v_n$ for
$n\in\N$. Then the \emph{infinite product} $\prod_{n\in\N}\fA_n$ is a
$\tau_{1,\ell}$-structure. It is obtained from the disjoint union of
the structures $\fA_n$ by identifying $v_n$ and $u_{n+1}$ for any
$n\in\N$. The only constant of this infinite product is~$u_0$.  If
$\fA\cong\fA_n$ for all $n\in\N$, then we write simply $\fA^\omega$
for the infinite product of the structures~$\fA_n$.

Standard applications of Ehrenfeucht-Fra\"\i{}ss\'e-games
(see~\cite{EbbF91}) yield:

\begin{proposition}\label{P-EF}
  Let $j,\ell,m\in\N$, $\fA_n,\fA_n'$ be $\tau_{2,\ell}$-structures for $n\in\N$
  and let $\fB,\fB'$ be some $\tau_{j+1,\ell}$-structures such that
  $\fA_n\equiv^\mch_m\fA'_n$ for $n\in\N$ and $\fB\equiv^\mch_m\fB'$. Then
  \[
    \fA_0\cdot\fB\equiv^\mch_m\fA_0'\cdot\fB' \quad\text{and}\quad
    \prod_{n\in\N}\fA_n\equiv^\mch_m\prod_{n\in\N}\fA_n'\ . 
  \]
\end{proposition}

\begin{remark}\label{R-typical-use}
  We sketch a typical use of the above proposition in this
  section. Let $x\in A^*$ be some sufficiently long
  word. Since $\equiv^\mch_m$ has only finitely many equivalence
  classes, there exist words $u,v,w$ with $x=uvw$ and
  $v\neq\varepsilon$ such that
  $(\fA^*_u,\{x\})\equiv^\mch_m(\fA^*_{uv},\{x\})$. Hence we obtain
  \[
    (\fA^*,\{x\}) =(\fA^*_{\varepsilon,u},\emptyset)\cdot
                      (\fA^*_{u},\{uvw\})
    \equiv_m^\mch(\fA^*_{\varepsilon,u},\emptyset)\cdot
                      (\fA^*_{uv},\{uvw\})
    \cong(\fA^*,\{uw\})\ .
  \]
  (This proves that every consistent property of a single element of
  $\fA^*$ is witnessed by some ``short'' word.) 

  The last isomorphism does not hold for the Muchnik-iteration since
  the clone predicate allows to express that the last letter of $u$
  and the first letter of $v$ are connected by some edge in the
  graph~$\fA$.
\end{remark}

\begin{convention}
  We consider complete deterministic finite automata
  $\cM=(Q,B,\iota,\delta,F)$, called \emph{automata} for short. Its
  language is denoted $L(\cM)$. We will also write $p.w$ for
  $\delta(p,w)$. The \emph{transition matrix} of $\cM$ is the tuple
  $T=(T_{p,q})_{p,q\in Q}$ with $T_{p,q}=\{b\in B\mid
  \delta(p,b)=q\}$. 
\end{convention}

As explained above, we will use automata to describe subsets of the
Shelah-Stupp iteration~$\fA^*$, i.e., the alphabet $B$ will always be
a finite subset of the universe of~$\fA$. These regular subsets have
the following nice property whose proof is obvious.

\begin{lemma}
  Let $\fA$ be a $\sigma$-structure with universe $A$ and let
  $\cM=(Q,B,\iota,\delta,F)$ be an automaton with alphabet $B\subseteq
  A$. Then, for any $u,v\in B^*$ with
  $\delta(\iota,u)=\delta(\iota,v)$, the mapping $f_{u,v}:uA^*\to
  vA^*:ux\mapsto vx$ is an isomorphism from $(\fA^*_u,L(\cM))$ onto
  $(\fA^*_v,L(\cM))$. 
\end{lemma}

As a consequence, the number of  isomorphism classes of
structures $(\fA^*_v,L(\cM))$ is finite. This fails in the Muchnik-iteration
even for $L(\cM)=\emptyset$: With $\fA=(\mathbb{N},\mathrm{succ})$ and
$m,n\in\mathbb N$, we have $(\fA_m^*,\clone)\cong(\fA_n^*,\clone)$ iff
$m=n$ since the structure $(\mathbb N,\mathrm{succ},m)$ can be defined
in $(\fA_m^*,\clone)$.

\subsection{Quantification}
\label{SS-quantification}
While multichains in the Shelah-Stupp-iteration can be rather
complicated, this section shows that, up to logical equivalence, we
can restrict attention to ``simple'' multichains. Here, ``simple''
means that they are regular and, even more, can be accepted by a
``small'' automaton.

\textit{For the rest of this section, let $\fA=(A,(R)_{R\in\sigma})$
  be some fixed $\sigma$-structure and $\ell,m\in\N$.  For $1\le
  i\le\ell$, let $\cM_i=(Q_i,B_i,\iota_i,F_i)$ be automata with
  $B_i\subseteq A$ such that $L(\cM_i)\subseteq A^*$ is a multichain
  in the Shelah-Stupp iteration $\fA^*$. Write $\bar L$ for the tuple
  of multichains $(L(\cM_1),\dots,L(\cM_\ell))$.}

\begin{proposition}\label{P-regular-chain-exists}
  Let $C\subseteq A^*$ be a chain. Then there exist $u,v\in A^*$,
  $E\subseteq\Down u$, and $F\subseteq \Down v$ such that
  $\iota_i.u=\iota_i.uv$ for all $1\le i\le\ell$ and
    $(\fA^*,\bar L,C)\equiv^\mch_m(\fA^*,\bar L,D)$
  with $D=E\cup uv^*F$.
\end{proposition}

\begin{proof}
  One shows the existence of $u_1\prec u_2\in A^*$ such that
  $C\cup\{u_1,u_2\}$ is a chain, $\iota_i.u_1=\iota_i.u_2$ for all
  $1\le i\le\ell$,
    $(\fA^*,\bar L) \cong  (\fA^*_{\varepsilon,u_1},\bar L)
          \cdot (\fA^*_{u_1,u_2},\bar L)^\omega$, and
    $(\fA^*,\bar L,C) \equiv^\mch_m
         (\fA^*_{\varepsilon,u_1},\bar L,C)
          \cdot (\fA^*_{u_1,u_2},\bar L,C)^\omega$.
  This uses arguments similar to those in Remark~\ref{R-typical-use}
  and Ramsey's theorem. The result follows with $u=u_1$, $uv=u_2$,
  $E=C\cap\Down u$, and $F=u^{-1}(C\cap\Down{u_2})$.
\end{proof}

The above proposition shows that every consistent property of a chain
is witnessed by some regular chain~$D$. Using the pigeonhole principle
and arguments as in Remark~\ref{R-typical-use}, one can bound the
lengths of $u$ and $v$ to obtain

\begin{proposition}\label{P-chain-quantification}
  Let $C\subseteq A^*$ be a chain. Then there exists an automaton
  $\cN$ with at most $2\prod_{1\le i\le\ell}|Q_i|\cdot T(\ell+1,m)$
  states such that $L(\cN)$ is a chain and $(\fA^*,\bar
  L,C)\equiv^\mch_m (\fA^*,\bar L,L(\cN))$.
\end{proposition}

It is our aim to prove a similar result for arbitrary multichains in
place of the chain $C$ in the proposition above. Certainly, in order
to get a small automaton for a multichain, the branching points of
this set have to be short words. Again using arguments as in
Remark~\ref{R-typical-use}, one obtains

\begin{lemma}\label{L-small-set-1}
  Let $M\subseteq A^*$ be a multichain. Then there exists a multichain
  $N\subseteq A^*$ such that
  \begin{itemize}
  \item $(\fA^*,\bar L,M)\equiv^\mch_m
    (\fA^*,\bar L,N)$ and
  \item any branching point of $N$ has length at most $k=\prod_{1\le
      i\le \ell}(|Q_i|+1)\cdot T(\ell+1,m)$. 
  \end{itemize}
\end{lemma}

\begin{lemma}\label{L-small-set-2}
  Let $M$ be a multichain such that all branching points of $M$ have
  length at most $s-1$. Then there exists an automaton $\cN$ with at
  most $(2\prod_{1\le i\le\ell}|Q_i|\cdot T(\ell+1,m))^{s+1}$ many
  states such that $L(\cN)$ is a multichain and $(\fA^*,\bar L,M)
  \equiv^\mch_m (\fA^*,\bar L,L(\cN))$.
\end{lemma}

\begin{proof}
  Let $n=\prod_{1\le i\le\ell}|Q_i|$ and $\bar
  L=(L(\cM_1),\dots,L(\cM_\ell))$.
  
  The lemma is shown by induction on~$s$. If $s=0$, then $M$ is a
  chain, i.e., the result follows from Prop.~\ref{P-chain-quantification}.

  Now let $M$ be a multichain such that any branching point has length
  at most~$s>0$. By the induction hypothesis, for every $a\in A$,
  there exists an automaton $\cN_a$ with at most $(2n
  T(\ell+1,m))^{s+1}$ many states such that $L(\cN_a)$ is a multichain
  and
  \[
    (\fA^*_a,\bar L,M)
    \equiv^\mch_m (\fA^*,a^{-1}L(\cM_1),\dots,a^{-1}L(\cM^a_\ell),L(\cN_a))\ . 
  \]
  
  Let $\theta$ be the equivalence relation on $A$ with
  $(a,b)\in\theta$ if and only if
  \begin{enumerate}
  \item $\delta_i(\iota_i,a)=\delta_i(\iota_i,b)$ for all $1\le
    i\le\ell$ and
  \item $(\fA^*_a,\bar L,M)\equiv^\mch_m(\fA^*_b,\bar L,M)$. 
  \end{enumerate}
  Let $H\subseteq A$ contain precisely one element~$h$ from any
  $\theta$-equivalence class.  Then the set $\bigcup\{a L(\cN_h)\mid
  a\,\theta\, h\in H\text{ and
  }a^{-1}M\neq\emptyset\}\cup(\{\varepsilon\}\cap M)$ is a multichain
  and can be accepted by some automaton~$\cN$ with the right number of
  states.
  
  Then $(\fA^*,\bar L,L(\cN))$ is obtained from $(\fA^*,\bar L,M)$ by
  replacing any subtree $(\fA^*_a,\bar L,M)$ with the equivalent
  structure $(\fA^*,a^{-1}L(\cM_1),\dots,a^{-1}L(\cM^a_\ell),L(\cN_h))$ for
  $a\,\theta h\in H$. Hence, by Prop.~\ref{P-EF},
  $(\fA^*,\bar L,M)\equiv^\mch_m(\fA^*,\bar L,L(\cN))$.
\end{proof}

Putting these two lemmas together, we obtain that, indeed, every
consistent property of a multichain $M$ is witnessed by some
multichain that can be accepted by some ``small'' automaton:

\begin{proposition}\label{P-multichain}
  Let $M\subseteq A^*$ be some multichain. Then there exists an
  automaton $\cN$ with at most $(2nT(\ell+1,m))^{s+1}$ many states
  (where $s=n\cdot T(\ell+1,m)$, $n=\prod_{1\le i\le\ell}|Q_i|$) such
  that $L(\cN)$ is a multichain and $(\fA^*,\bar L,M) \equiv^\mch_m
  (\fA^*,\bar L,L(\cN))$.
\end{proposition}

Now a result analogous to Rabin's basis theorem follows immediately

\begin{theorem}
  Let $\fA$ be a $\sigma$-structure, let $\varphi$ be an $\MSO^\mch$-formula
  in the language of the Shelah-Stupp-iteration $\fA^*$ with free variables
  $X_1,\dots,X_\ell$ and let $L_1,\dots,L_\ell\subseteq A^*$ be regular
  languages such that $(\fA^*,L_1,\dots,L_\ell)\models^\mch\varphi$. Then
  $(\fA^*,L_1,\dots,L_\ell)\models^{\mathrm{reg-}\mch}\varphi$ where
  $\models^{\mathrm{reg-}\mch}$ denotes that set quantification is restricted
  to regular multichains. 
\end{theorem}

Recall that Rabin's basis theorem follows from his tree theorem whose
proof, in turn, uses the effective complementation of Rabin tree
automata. While the above theorem is an analogue of Rabin's basis
theorem, the proof is more direct and does in particular not rest on
any complementation of automata.

\section{Shelah-Stupp-iteration is $(\MSO^\mch,\MSO^\w)$-compatible}
\label{S-Shelah-Stupp}
The results of the previous section, as explained at the beginning,
imply that quantification in an $\MSO^\mch$-sentence can be restricted to
regular sets that are accepted by ``small'' automata. In this section,
we will use this insight to reduce the $\MSO^\mch$-theory of the 
Shelah-Stupp-iteration to the $\MSO^\w$-theory of the base structure. 

Fix some $\sigma$-structure $\fA$ with universe~$A$, some finite set
of states~$Q$, some initial state~$\iota$, and some set of final
states $F\subseteq Q$. Then, for any automaton
$\cM=(Q,B,\iota,\delta,F)$ with $B\subseteq A$, the language $L(\cM)$
is a set in the Shelah-Stupp-iteration $\fA^*$ while its transition
matrix is a tuple of finite sets in the base structure
$\fA$. \emph{The idea of our reduction is that $\MSO^\mch$-properties
  of the set $L(\cM)$ in the Shelah-Stupp-iteration~$\fA^*$ can
  (effectively) be translated into $\MSO^\w$-properties of the
  transition matrix~$T$ in the base structure~$\fA$. }

In precisely this spirit, the following lemma expresses simple
properties of the automaton $\cM$ and of the language $L(\cM)$ in
terms of $\FO$-properties of~$(\fA,T)=(\fA,(T_{p,q})_{p,q\in Q})$. 

\begin{lemma}\label{L-simple-formulas}
  Let $F\subseteq Q$ be finite sets and $\iota\in Q$. There exist
  formulas $\dkpath_{(Q,p,q)}$ for $p,q\in Q$ and
  $\mathrm{mchain}_{(Q,\iota,F)}$ of $\FO$ with free variables
  $T_{p,q}$ for $p,q\in Q$ such that for any $\sigma$-structure $\fA$
  and any automaton $\cM=(Q,B,\iota,\delta,F)$ with transition
  matrix~$T$:
  \begin{enumerate}[(1)]
  \item $(\fA,T)\models^\w \dkpath_{(Q,p,q)}$ iff there exists a word
    $w\in A^*$ with $\delta(p,w)=q$. 
  \item $(\fA,T)\models^\w\mathrm{mchain}_{(Q,\iota,F)}$ iff $L(\cM)$ is a
    multichain. 
  \end{enumerate}
\end{lemma}

\begin{proof}
  The proof is based on the observation that (1) one only needs to
  search for a path of length at most $|Q|$ and (2) that $L(\cM)$ is a
  multichain iff no branching point belongs to some cycle.
\end{proof}

So far, we showed that simple properties of $L(\cM)$ are actually
$\FO$- (and therefore $\MSO^\w$-) properties of the transition matrix
of $\cM$. We now push this idea further and consider arbitrary
$\MSO^\mch$-properties of a tuple of languages
$L(\cM_1),\dots,L(\cM_\ell)$.

\begin{theorem}\label{T-reduction}
  There is an algorithm with the following specification\\
  {\rm input:}
  \begin{minipage}[t]{.9\linewidth}
    \begin{itemize}
    \item $\ell\in\N$,
    \item finite sets $F_i\subseteq Q_i$ and states $\iota_i\in Q_i$
      for $1\le i\le\ell$,
    \item and a formula $\alpha$ with free variables among
      $L_1,\dots,L_\ell$ in the language of the Shelah-Stupp-iteration
      $\fA^*$. 
    \end{itemize}
  \end{minipage}

  \noindent
  {\rm output:}
  \begin{minipage}[t]{.9\linewidth}
    A formula $\alpha_{(\bar Q,\bar\iota,\bar F)}$ in the language of
    $\fA$ with free variables among $T^i_{p,q}$ for $p,q\in Q_i$ and
    $1\le i\le\ell$ with the following property:

    If $\fA$ is a $\sigma$-structure and
    $\cM_i=(Q_i,B_i,\iota_i,T^i,F_i)$ are automata with $B_i\subseteq
    A$ for $1\le i\le\ell$, then
    \[
      (\fA^*,L(\cM_1),L(\cM_2),\dots,L(\cM_\ell))\models^\mch\alpha
      \iff
      (\fA,T^1,T^2,\dots,T^\ell)\models^\w
             \alpha_{(\bar Q,\bar\iota,\bar F)}\ .
    \]
  \end{minipage}
\end{theorem}

\begin{proof}
  The proof proceeds by induction on the construction of the formula
  $\alpha$, we only sketch the most interesting part $\alpha=\exists
  X\,\beta$.  Set $n=\prod_{1\le i\le\ell}|Q_i|$, $s=nT(\ell+1,m)$,
  and $k=(2nT(\ell+1,m))^{s+1}$. Let $\fA$ be a $\sigma$-structure and
  let $\cM_i=(Q_i,B_i,\iota_i,\delta_i,F_i)$ be automata with
  $B_i\subseteq A$ and transition matrix~$T^i$. Then, by
  Prop.~\ref{P-multichain},
  $(\fA^*,L(\cM_1),\dots,L(\cM_\ell))\models^\mch\alpha$ iff there exists
  an automaton~$\cN$ with~$k$ states such that
  \[
    (\fA^*,L(\cM_1),L(\cM_2),\dots,L(\cM_\ell),L(\cN))\models^\mch\beta\ .
  \]
  Using the induction hypothesis on $\beta$ and $\beta_{(\bar
    Q,\bar\iota,\bar F)}$, this is the case if and only if there exist finite
  sets $T^{\ell+1}_{i,j},B\subseteq A$ for $i,j\in[k]=\{1,2,\dots,k\}$
  such that
  \begin{itemize}
  \item $T^{\ell+1}$ forms the transition matrix of some automaton 
    with alphabet $B$
  \item for some $F\subseteq[k]$, the automaton
    $\cM_{\ell+1}=([k],B,1,T^{\ell+1},F)$
    \begin{itemize}
    \item accepts a multichain $M$ (i.e.,
      $(\fA,T^{\ell+1})\models^\w\textrm{mchain}_{([k],1,F)}$) and
    \item this multichain satisfies $\beta$ (i.e.,
      $\fA,T^1,\dots,T^{\ell+1}\models^\w\beta_{((\bar
        Q,[k]),(\bar\iota,1),(\bar F,F))}$). 
    \end{itemize}
  \end{itemize}
  Since all these properties can be expressed in $\MSO^\w$, the
  construction of $\alpha_{(\bar Q,\bar\iota,\bar F)}$ is complete. 
\end{proof}

As an immediate consequence, we get a uniform version of Shelah and
Stupp's theorem for the logics $\MSO^\w$ and $\MSO^\mch$:

\begin{theorem}\label{T-Stupp-iteration}
  Finitary Shelah-Stupp-iteration is $(\MSO^\mch,\MSO^\w)$-compatible. 
\end{theorem}

\begin{remark}
  $(\MSO^\ch,\FO)$-compatibility of Shelah-Stupp-iteration
  \cite{KusL06a} can alternatively be shown along the same lines: One
  allows incomplete automata and proves an analogue of
  Prop.~\ref{P-chain-quantification} for the logic $\MSO^\ch$. Then
  Theorem~\ref{T-Stupp-iteration} can be shown for the pair of logics
  $(\MSO^\ch,\FO)$.
\end{remark}

\section{Infinitary Muchnik-iteration is not $(\FO,\MSO^\w)$-compatible}
\label{S-Muchnik}
Our argument goes as follows: From a set $M\subseteq\N$, we construct
a tree~$\fA_M$. The $\MSO^\w$-theory of this tree will be independent
from~$M$ and $M$ will be $\FO$-definable in the infinitary
Muchnik-iteration $(\fA^\infty_M,\clone)$. Assuming
$(\FO,\MSO^\w)$-compatibility of the infinitary Muchnik-iteration, the
set~$M$ will be reduced uniformly to the $\MSO^\w$-theory of
$\fA_M$. For $M\neq N$, this yields a contradiction.

A \emph{tree} is a structure $(V,\preceq,r)$ where $\preceq$ is a
partial order on $V$ such that, for any $v\in V$, $(\down v,\preceq)$
is a finite linear order and $r\preceq v$ for all $v\in V$.

We will consider the set
$T_\omega=\{(a_1,m_1)(a_2,m_2)\dots(a_k,m_k)\in(\N\times\N)^*\mid
m_1>m_2>m_3\dots >m_k\}$ of sequences in $\N^2$ whose second
components decrease. This set, together with the prefix relation
$\preceq$, forms a tree $(T_\omega,\preceq,\varepsilon)$ with root
$\varepsilon$ that we also denote $T_\omega$. Nodes of the form
$w(a,0)$ are leaves of $T_\omega$. Any inner node of $T_\omega$ has
infinitely many children (among them, there are infinitely many
leaves). Furthermore, all the branches of $T_\omega$ are finite. Even
more, if $x$ is a node different from the root, then the branches
passing through $x$ have bounded length.

We will also consider the set $T_\infty=a^*T_\omega$ where $a$ is an
arbitrary symbol. Together with the prefix relation, this yields
another tree $(T_\infty,\preceq,\varepsilon)$ that we
denote~$T_\infty$. Differently from $T_\omega$, it has an infinite
branch, namely the set of all nodes $a^n$ for $n\in\N$.

For two trees $S$ and $T$ and a node $v$ of $S$, let $S\cdot_v T$
denote the tree obtained from the disjoint union of $S$ and $T$ by
identifying $v$ with the root of $T$ (i.e., the node $v$ gets
additional children, namely the children of the root in $T$).

It is important for our later arguments that this operation transforms
trees equivalent wrt.~$\equiv_m^\w$ into equivalent structures.  More
precisely

\begin{proposition}\label{P-EF2}
  Let $S$, $T$, and $T'$ be trees and $k\in\N$ such that 
  $T\equiv_k^\w T'$. Then $S\cdot_v T\equiv^\w_k S\cdot_v T'$ for any
  node $v$ of~$S$. 
\end{proposition}

With $a^{\le n}=\{\varepsilon, a,a^2,\dots,a^n\}$, the set $a^{\le
  n}T_\omega$ together with the prefix relation and the root, is
considered as a tree that we denote $a^{\le n}T_\omega$.

\begin{proposition}\label{P-equivalent}
  For any $k\in\N$, we have $T_\omega\equiv^\w_k T_\infty$. 
\end{proposition}

\begin{proof}
  The statement is shown by induction on~$k$ where the base case $k=0$
  is trivial. To show $T_\omega\equiv^\w_{k+1} T_\infty$, it suffices
  to prove for any formula $\varphi(X)$ of quantifier-depth at most~$k$
  \[
     T_\omega\models^\w\exists X\,\varphi(X)\iff
     T_\infty\models^\w\exists X\,\varphi(X)\ . 
  \]
  Assuming $T_\infty\models^\w\exists X\,\varphi$, there exist
  $n\in\N$ and $M\subseteq a^{\le n}T_\omega$ finite with
  $(T_\infty,M)\models^\w\varphi$. Hence we have
  \begin{align*}
    (T_\infty,M) 
      &\cong (a^{\le n}T_\omega,M)\cdot_{a^n} (T_\infty,\emptyset)\\
      &\equiv^\w_k 
             (a^{\le n}T_\omega,M)\cdot_{a^n} (T_\omega,\emptyset)
         \text{ by Prop.\ \ref{P-EF2} and the induction hypothesis}\\
      &\cong (a^{\le n}T_\omega,M)\ . 
  \end{align*}
  Hence $(a^{\le n}T_\omega,M)\models^\w\varphi$ and therefore $a^{\le
    n}T_\omega\models^\w\exists X\,\varphi$. Using
  $T_\omega\equiv^\w_{k+1}a^{\le n}T_\omega$ (see complete paper for
  the proof), we obtain $T_\omega\models^\w\exists X\,\varphi$.

  Conversely, one can argue similarly again using
  $T_\omega\equiv^\w_{k+1}a^{\le n}T_\omega$.
\end{proof}

\begin{remark}
  This proves that the existence of an infinite path cannot be
  expressed in weak monadic second order logic since $T_\infty$ has
  such a path and $T_\omega$ does not.
\end{remark}

Using an idea from \cite{CouW98}, the existence of an infinite path is
a first-order property of the infinitary Muchnik-iteration. The
following lemma pushes this idea a bit further:

\begin{lemma}\label{L-U}
  Let $T=(T,\le,r)$ be a tree and let $U\subseteq T$ be the union of
  all infinite branches of $T$. Then the $\MSO^\w$-theory
  of $(T,\le,r,U)$ is uniformly reducible to the $\MSO^\w$-theory of
  the infinitary Muchnik-iteration $(T^\infty,\clone)$ of the
  tree~$(T,\le,r)$ without the extra predicate. 
\end{lemma}

For $M\subseteq\N$, let $A_M=\{b^m\mid m\in M\}T_\infty\cup \{b^m\mid
m\notin M\}T_\omega$ and $\fA_M=(A_M,\preceq,\varepsilon)$. Then
$\fA_M$ is obtained from the linear order $(\mathbb
N,\le)\cong(b^*,\preceq)$ by attaching the tree $T_\infty$ to elements
from~$M$ and the tree $T_\omega$ to the remaining numbers.

\begin{theorem}\label{T-complicated}
  For $M\subseteq\N$, we have $\fA_M\equiv^\w_k T_\omega$ for all
  $k\in\N$, and $M$ can be reduced to the $\FO$-theory of the
  infinitary Muchnik-iteration $(\fA_M^\infty,\clone)$.
\end{theorem}

\begin{proof}
  Using Ehrenfeucht-Fra\"\i{}ss\'e-games and Prop.~\ref{P-equivalent},
  one obtains
  \[
    \fA_M\equiv^\w_k (b^* T_\omega,\preceq,\varepsilon)
    \cong T_\infty\equiv^\w_k T_\omega\ . 
  \]

  For the second statement, it suffices, by Lemma~\ref{L-U}, to
  reduce~$M$ to the first-order theory of
  $(A_M,\preceq,\varepsilon,U)$ where $U=b^*\cup \{b^m\mid m\in
  M\}a^*$ is the set of nodes of the tree $\fA_M$ that belong to some
  infinite branch.
\end{proof}

If a transformation $t$ is $(\FO,\MSO^\w)$-compatible, then for any
structure~$\fA$, the $\FO$-theory of~$t(\fA)$ can be reduced to the
the $\MSO^\w$-theory of~$\fA$. Contrary to this, the above theorem
states that the $\FO$-theory of the infinitary Muchnik-iteration can
be arbitrarily more complicated than the $\MSO^\w$-theory of the base
structure. Hence we obtain

\begin{corollary}
  Infinitary Muchnik-iteration is not $(\FO,\MSO^\w)$-compatible.
\end{corollary}

\section{Summary}
\label{S-summary}
Table~\ref{tab:sum} summarizes our knowledge about the compatibility
of Muchnik's and Shelah \& Stupp's iteration. It consists of four
subtables dealing with finitary and infinitary Muchnik-iteration and
with finitary and infinitary Shelah-Stupp-iteration. The sign + in
cell $(\mathcal K,\mathcal L)$ of a subtable denotes that the
respective iteration is $(\mathcal K,\mathcal L)$-compatible, --
denotes the opposite. Minus-signs without further marking hold since
the base structure can be defined in any of its iterations. Capital
letters denote references: (A)~is \cite{She75}, (B)~\cite{Wal02},
(C)~\cite[Prop.~3.4]{KusL06a}, (D)~\cite[Thm.~4.10]{KusL06a},
(E)~Theorem~\ref{T-Stupp-iteration}, (F)~Theorem~\ref{T-complicated},
and (G) since the base structure is definable in its iteration and
finiteness of a set is no MSO-property.  Small letters denote that the
result follows from Theorem~\ref{T-fin-inf} below and some further
``simple'' arguments from the result marked by the corresponding
capital letter.

\begin{theorem}\label{T-fin-inf}
  Let $(\mathcal K,\mathcal L)$ be any of the pairs of logics
  $(\MSO^\cl,\MSO)$, $(\MSO^\ch,\MSO^\ch)$, or
  $(\MSO^\mch,\MSO^\mch)$.  There exists a computable function $\red$
  such that, for any $\sigma$-structure $\fA$, $\red$ reduces the
  $\mathcal K$-theory of $(\fA^\infty,\clone)$ to the $\mathcal
  L$-theory of $(\fA^*,\clone)$.

  The same holds for the Shelah-Stupp-iterations. 
\end{theorem}

The two questions marks in Table~\ref{tab:sum} express that it is not
clear whether finitary Muchnik-iteration is $\MSO^\w$-compatible or not.

Note the main difference between Muchnik- and Shelah-Stupp-iteration:
the latter is $\mathcal K$-com\-patible for all relevant logics while
only $\MSO$ behaves that nicely with respect to (infinitary)
Muchnik-iteration

A referee proposed to also consider the variant of MSO where set
quantification is restricted to countable sets. As to whether Muchnik
iteration is compatible with this logic is not clear at the moment.
\begin{table}
  \newlength{\ppp}\settowidth{\ppp}{+}
  \newcommand{\p}{{\makebox[\ppp]{+}}\xspace}
  \newcommand{\m}{{\makebox[\ppp]{\rule[-2pt]{0pt}{8pt}--}}\xspace}
  \centering
  \begin{tabular}[t]{l|@{}c@{}c@{}c@{}||l|@{}c@{}c@{}c}
  \setlength{\columnsep}{-10pt}
              & \multicolumn{3}{c||}{Muchnik}&&
              \multicolumn{3}{c}{inf.\ Muchnik}\\
              & $\;\MSO\;$ & $\MSO^\w\;$ & $\;\FO\;$   && $\;\MSO\;$
              &$\MSO^\w\;$& $\FO\;$ \\ 
                  \hline
  $\MSO$  & \p (B)&\m&\m&$\MSO^\cl$&\p (b)&\m&\m\\
  $\MSO^\w$   &\m (g)&{{?}} &\m&$\MSO^\w$&\m (g)&\m (f)& \m \\
  $\FO$       &\p (b)&{{?}} &\m (C)\;&$\FO$&\p{ (b)}&\m (F)&\m (c)\\ 
         \hline\hline
       & \multicolumn{3}{c||}{Shelah-Stupp}&&
              \multicolumn{3}{c}{inf.\ Shelah-Stupp}\\
              & $\;\MSO\;$ & $\MSO^\w\;$ & $\;\FO\;$   && $\;\MSO\;$
              &$\MSO^\w\;$& $\FO\;$ 
        \\\hline
  $\MSO$  &\p (A)&\m &\m &$\MSO^\cl$&\p (a)&\m & \m\\
  $\MSO^\mch$ &\m (g)&\p (E)&\m &$\MSO^\mch$&\m (g)&\p (e)&\m \\
  $\MSO^\w$   &\m (G) &\p (e)&\m&$\MSO^\w$&\m (g)&\p (e)& \m\\
  $\MSO^\ch$  &\p (a)&\p (e)&\p (D)&$\MSO^\ch$&\p (a)&\p (e)&\p (d)\\
  $\FO$       &\p (a)&\p (e)&\p (d)&$\FO$&\p (a)&\p (e)&\p (d)
  \end{tabular}
  \caption{summary}
  \label{tab:sum}
\end{table}


\begin{thebibliography}{10}

\bibitem{BluK05}
A.~Blumensath and S.~Kreutzer.
\newblock An extension of {M}uchnik's theorem.
\newblock {\em Journal of Logic and Computation}, 15:59--64, 2005.

\bibitem{Cau02b}
D.~Caucal.
\newblock On infinite terms having a decidable monadic theory.
\newblock In {\em MFCS'02}, Lecture Notes in Comp.\ Science vol.\ 2420, pages
  165--176. Springer, 2002.

\bibitem{ColL07}
Th. Colcombet and Ch. L{\"o}ding.
\newblock Transforming structures by set interpretations.
\newblock {\em Logical Methods in Computer Science}, 3:1--36, 2007.

\bibitem{Cou94}
B.~Courcelle.
\newblock Monadic second-order definable graph transductions: a survey.
\newblock {\em Theoretical Computer Science}, 126:53--75, 1994.

\bibitem{CouW98}
B.~Courcelle and I.~Walukiewicz.
\newblock Monadic second-order logic, graph coverings and unfoldings of
  transition systems.
\newblock {\em Ann.\ Pure Appl.\ Logic}, 92(1):35--62, 1998.

\bibitem{EbbF91}
H.-D. Ebbinghaus and J.~Flum.
\newblock {\em Finite Model Theory}.
\newblock Springer, 1991.

\bibitem{FefV59}
S.~Feferman and R.L. Vaught.
\newblock The first order properties of algebraic systems.
\newblock {\em Fund.\ Math.}, 47:57--103, 1959.

\bibitem{GurS83}
Y.~Gurevich and S.~Shelah.
\newblock Rabin's uniformization problem.
\newblock {\em J.\ of Symb.\ Logic}, 48:1105--1119, 1983.

\bibitem{KusL06b}
D.~Kuske and M.~Lohrey.
\newblock Logical aspects of {C}ayley-graphs: {T}he monoid case.
\newblock {\em International Journal of Algebra and Computation}, 16:307--340,
  2006.

\bibitem{KusL06a}
D.~Kuske and M.~Lohrey.
\newblock Monadic chain logic over iterations and applications to push-down
  systems.
\newblock In {\em LICS 2006}, pages 91--100. IEEE Computer Society, 2006.

\bibitem{Rab69}
M.O. Rabin.
\newblock Decidability of second-order theories and automata on infinite trees.
\newblock {\em Trans.\ Amer.\ Math.\ Soc.}, 141:1--35, 1969.

\bibitem{Rab72}
M.O. Rabin.
\newblock {\em Automata on infinite objects and {C}hurch's problem}.
\newblock American Mathematical Society, Providence, R.I., 1972.
\newblock Conference Board of the Mathematical Sciences Regional Conference
  Series in Mathematics, No. 13.

\bibitem{See91}
D.~Seese.
\newblock The structure of models of decidable monadic theories of graphs.
\newblock {\em Annals of Pure and Applied Logic}, 53:169--195, 1991.

\bibitem{Sem84}
A.L. Semenov.
\newblock Decidability of monadic theories.
\newblock In M.~Chytil and V.~Koubek, editors, {\em MFCS'84}, Lecture Notes in
  Comp.\ Science vol.\ 176, pages 162--175. Springer, 1984.

\bibitem{She75}
S.~Shelah.
\newblock The monadic theory of order.
\newblock {\em Annals of Mathematics}, 102:379--419, 1975.

\bibitem{Tho87}
W.~Thomas.
\newblock On chain logic, path logic, and first-order logic over infinite
  trees.
\newblock In {\em LICS'87}, pages 245--256. IEEE Computer Society Press, 1987.

\bibitem{Wal02}
I.~Walukiewicz.
\newblock Monadic second-order logic on tree-like structures.
\newblock {\em Theoretical Computer Science}, 275:311--346, 2002.

\end{thebibliography}
\end{document}